\documentclass[letterpaper, 10 pt, conference]{ieeeconf}  %

\IEEEoverridecommandlockouts                              %
\usepackage{tikz}
\usepackage{subcaption}
\usepackage{booktabs}
\usepackage{url}
\usepackage{cite}
\usepackage{amsmath,amssymb,amsfonts}
\usepackage{algorithmic}
\usepackage{graphicx}
\usepackage{textcomp}
\usepackage{xcolor}

\def\BibTeX{{\rm B\kern-.05em{\sc i\kern-.025em b}\kern-.08em
    T\kern-.1667em\lower.7ex\hbox{E}\kern-.125emX}}
\newtheorem{thm}{Theorem }%
\newtheorem{prop}{Proposition}%
\newtheorem{cor}{Corollary}%
\newtheorem{lem}{Lemma }%
\newtheorem{rem}{Remark }%
\newtheorem{ass}{Assumption}%

\usepackage[ruled,vlined,linesnumbered]{algorithm2e}
\usepackage{makecell}

\newcommand{\Real}{\mathbb R}
\newcommand{\eps}{\varepsilon}
\newcommand{\abs}[1]{\left\vert#1\right\vert}

\newcommand{\U}{\mathcal{U}}

\let\subset\subseteq

\usepackage{xcolor}

\title{\LARGE \bf
Verifiable Error Bounds for Physics-Informed Neural Network Solutions of Lyapunov and Hamilton--Jacobi--Bellman Equations
}

\author{Jun Liu  %
\thanks{This research was supported in part by the NSERC of Canada and the Canada Research Chairs program. The author is with the Department of Applied Mathematics, Faculty of Mathematics, University of Waterloo, Waterloo, Ontario N2L 3G1, Canada.  Email: \texttt{j.liu@uwaterloo.ca}}%
}

\begin{document}

\maketitle
\pagestyle{plain}

\begin{abstract}
Many core problems in nonlinear systems analysis and control can be recast as solving partial differential equations (PDEs) such as Lyapunov and Hamilton--Jacobi--Bellman (HJB) equations. Physics-informed neural networks (PINNs) have emerged as a promising mesh-free approach for approximating their solutions, but in most existing works there is no rigorous guarantee that a small PDE residual implies a small solution error. This paper develops verifiable error bounds for approximate solutions of Lyapunov and HJB equations, with particular emphasis on PINN-based approximations. For both the Lyapunov and HJB PDEs, we show that a verifiable residual bound yields relative error bounds with respect to the true solutions as well as computable \emph{a posteriori} estimates in terms of the approximate solutions. For the HJB equation, this also yields certified upper and lower bounds on the optimal value function on compact sublevel sets and quantifies the optimality gap of the induced feedback policy. We further show that one-sided residual bounds already imply that the approximation itself defines a valid Lyapunov or control Lyapunov function. We illustrate the results with numerical examples.

\end{abstract}

\section{Introduction}

Many core problems in nonlinear systems analysis and control admit partial differential equation (PDE) characterizations. 
For example, Lyapunov functions for stability analysis can be characterized through first-order PDEs associated with the system dynamics \cite{camilli2001generalization,giesl2007construction,liu2025physics}, while optimal value functions for nonlinear feedback control are characterized by Hamilton--Jacobi--Bellman (HJB) equations \cite{bardi1997optimal,beard1997galerkin}. 
Solving such PDEs therefore provides a direct route to stability certificates, performance analysis, and feedback synthesis.

Despite their importance, these PDEs are notoriously difficult to solve. Classical grid-based numerical methods often become computationally prohibitive due to the curse of dimensionality and typically provide only discretized approximations that do not readily yield certificates for stability or optimality \cite{bardi1997optimal,mitchell2005time,beard1997galerkin}. 
This challenge has motivated growing interest in neural-network-based PDE solvers, especially physics-informed neural networks (PINNs), which approximate the solution by training a neural network to satisfy the governing PDE through residual minimization at sampled collocation points \cite{raissi2019physics,lagaris1998artificial}. 
PINNs and related neural-network methods have recently been used to compute approximate solutions of Lyapunov and HJB equations \cite{liu2023towards,liu2025physics,meng2024physics,furfaro2022physics,fotiadis2025physics} and learn Lyapunov and control Lyapunov functions \cite{chang2019neural,grune2021computing,gaby2022lyapunov,zhou2022neural,liu2023towards,zhou2024physics,liu2025physics,liu2025formally}.

These methods train neural networks by minimizing PDE residuals evaluated at sampled collocation points. 
A small residual is therefore often interpreted as evidence that the learned approximation is accurate. However, this implication is far from automatic. This raises a natural question: 
\begin{quote}
\emph{How does a pointwise residual bound translate into a true solution error bound for the PDE?}
\end{quote}

Answering this question is particularly important in control applications. For instance, when computing Lyapunov functions, one must ensure that the approximation provides a valid stability certificate. Similarly, for HJB equations, one would like guarantees on the optimal value function and on the performance of the feedback law induced by the approximation. In most existing PINN-based approaches, however, the residual error is used only as a training objective, and there is generally no rigorous relationship between the residual magnitude and the error in the computed solution.

The main goal of this paper is to establish such a relationship for stationary Lyapunov and HJB equations. 
Given an approximate solution, possibly produced by a PINN, we derive residual conditions under which the true solution error can be rigorously bounded. 
For both the Lyapunov and HJB PDEs, we propose a remarkably simple relative pointwise residual bound with respect to the stage cost and show that, when verified over the entire domain of interest, it yields clean relative error bounds with respect to the true solution, as well as easily computable \emph{a posteriori} bounds in terms of the approximate solution on the same domain.

More specifically, the main contributions of this paper are:
\begin{enumerate}
\item Certified error bounds for Lyapunov PDE solutions derived from pointwise relative residual bounds.
\item Certified upper and lower bounds for HJB value functions, which also yield rigorous guarantees on the optimality gap of the induced feedback policy.
\item One-sided pointwise residual conditions that directly certify Lyapunov or control Lyapunov functions.
\item Practically verifiable sufficient conditions for pointwise residual bounds based on local Hessian estimates.
\end{enumerate}

Although motivated by PINN-based approximations, the analysis in this paper is not tied to a particular training method or network architecture. 
The results apply to any approximate solution satisfying the stated residual conditions. 
To the best of our knowledge, this is the first work to derive formally verifiable solution error bounds for neural approximations of Lyapunov and HJB equations.

\section{Preliminaries}\label{sec:prelim}

Consider the autonomous system
\begin{equation}
    \label{eq:sys}
    \dot x = f(x),
\end{equation}
and the control-affine system
\begin{equation}
    \label{eq:control_sys}
    \dot x = f(x) + g(x)u,
\end{equation}
where \(f:\Real^n\to\Real^n\) and \(g:\Real^n\to\Real^{n\times k}\), with state \(x\in\Real^n\) and control \(u\in\Real^k\). It is assumed that $f(0)=0$, so both system \eqref{eq:sys} and \eqref{eq:control_sys} have an equilibrium at the origin in the absence of control. We assume that for each initial condition in the domain of interest and each admissible control input, a unique state trajectory exists. 
We write \(\phi(t,x)\) for a trajectory of \eqref{eq:sys} starting from \(x\), and \(\phi(t,x,u)\) for a trajectory of \eqref{eq:control_sys} starting from \(x\) under the input \(u(\cdot)\).

In this paper, we focus on two types of PDEs:

\textbf{Lyapunov PDE (stability certificates on a prescribed domain):}
For \eqref{eq:sys} with an asymptotically stable equilibrium at \(x=0\), a Lyapunov function on a set \(\Omega\subset\Real^n\) containing the origin can be characterized by the first-order linear PDE
\begin{equation}
\label{eq:lyap_pde}
DV(x)\cdot f(x) = -\omega(x),\quad x\in\Omega,
\end{equation}
together with the normalization \(V(0)=0\), where $DV$ denotes the gradient of $V$ and 
\(\omega:\Real^n\to\Real_{\ge0}\) is positive definite. 
Thus, computing \(V\) amounts to solving the PDE (\ref{eq:lyap_pde}). Existence and regularity of solutions to \eqref{eq:lyap_pde} are discussed later in Lemma~\ref{lem:lyap} and its accompanying footnote, in Section \ref{sec:lyap}. 
The resulting solution, or a sufficiently accurate smooth approximation of it, can provide a Lyapunov certificate, and its sublevel sets can be used to verify inner estimates of the domain of attraction of $x=0$ for (\ref{eq:sys}) defined by
\[
\mathcal{D} := \{x\in\Real^n:\ \lim_{t\to\infty}\|\phi(t,x)\|=0\}. 
\]

\textbf{HJB equation (optimal value and feedback control synthesis):}
For the control-affine system (\ref{eq:control_sys}), consider the infinite-horizon cost
\begin{equation}
\label{eq:cost}
J(x,u)=\int_0^\infty \Big(Q(\phi(t,x,u)) + u(t)^T\,R(\phi(t,x,u))\,u(t)\Big)\,dt,
\end{equation}
where \(Q:\Real^n\to\Real_{\ge0}\) and \(R:\Real^n\to\Real^{k\times k}\). 
Let \(\U\) denote the set of admissible controls that stabilize the system with finite cost. 
The optimal value function is
\begin{equation}\label{eq:optimal_value}
V^*(x)=\inf_{u\in\U} J(x,u).    
\end{equation}
It is characterized by the Hamilton--Jacobi--Bellman (HJB) PDE
\begin{equation}
\label{eq:hjb_inf}
0=\inf_{u\in\Real^k}\left\{Q+u^T R u + DV\cdot\big(f+gu\big)\right\},\quad x\in\Omega,
\end{equation}
together with the normalization \(V(0)=0\), where the dependence on \(x\) is omitted for notational simplicity. 
Since the expression inside the braces is quadratic in \(u\), the minimizer is
\begin{equation}
\label{eq:hjb_u_star}
u^*(x)= -\tfrac12 R^{-1}(x)\,g^T(x)\,DV^*(x)^T. 
\end{equation}
Substituting this into (\ref{eq:hjb_inf}) yields the equivalent nonlinear PDE
\begin{equation}
\label{eq:hjb_reduced}
0 = Q + DV\cdot f
-\tfrac14\, DV\,g\,R^{-1}\,g^T\,DV^T,\quad x\in\Omega,
\end{equation}
again with \(V(0)=0\).

\section{Error bounds for the Lyapunov PDE}\label{sec:lyap}

For the Lyapunov PDE \eqref{eq:lyap_pde}, consider the representation
\begin{equation}
\label{eq:lyap_integral}
V(x)=\int_0^\infty \omega(\phi(t,x))\,dt,
\end{equation}
where $\phi(t,x)$ denotes the trajectory of \eqref{eq:sys} starting from $x$. %

\begin{ass}
\label{ass:lyap}
The origin is asymptotically stable for system \eqref{eq:sys}. Let $\mathcal D$ be the domain of attraction for the origin and $\Omega\subset\mathcal{D}$ be forward invariant and contain the origin in its interior. Assume that $f$ is locally Lipschitz, $\omega$ is continuous, and the improper integral in \eqref{eq:lyap_integral} converges for all $x$ in some neighborhood of the origin
\footnote{One sufficient condition for this to hold is that $\omega$ is Lipschitz around the origin and the origin is
exponentially stable (see \cite[Remark 1]{liu2025physics}).}. 
\end{ass}

We state the following result from \cite[Proposition~2]{liu2025physics}.

\begin{lem}\label{lem:lyap}
Under Assumption~\ref{ass:lyap}, the function $V$ defined by \eqref{eq:lyap_integral} is the unique continuous solution to the Lyapunov PDE \eqref{eq:lyap_pde} on $\Omega$ satisfying $V(0)=0$ in the viscosity sense\footnote{Under the additional assumption that $f$ is continuously differentiable, $A=Df(0)$ is Hurwitz, and $\omega$ is continuously differentiable, it is shown in \cite[Proposition~2]{liu2025physics} that $V$ is also the unique continuously differentiable solution to (\ref{eq:lyap_pde}) on $\Omega$. We also note that the forward invariance of $\Omega$ was implicitly assumed in \cite[Proposition~2]{liu2025physics} and is stated explicitly here.}.
\end{lem}

Let $\hat V$ be an approximate solution to \eqref{eq:lyap_pde} with residual
\begin{equation}\label{eq:lyap_residual}
r(x):=D\hat V(x)\cdot f(x)+\omega(x).    
\end{equation}
We have the following result. 

\begin{thm} 
\label{thm:lyap_error}
Let Assumption~\ref{ass:lyap} hold and $\hat V\in C^1(\Omega)$ satisfy $\hat V(0)=0$. 
Assume that the residual $r$ in \eqref{eq:lyap_residual} satisfies 
\begin{equation}\label{eq:relative_residual_lyap}
|r(x)|\le \varepsilon\,\omega(x),\qquad \forall x\in\Omega,
\end{equation}
for some $\varepsilon \in [0,1)$. Then for all $x\in\Omega$,
\begin{equation}\label{eq:relative_error}
|\hat V(x)-V(x)|\le \varepsilon\,V(x),
\end{equation}
\begin{equation}\label{eq:posterior_error}
|\hat V(x)-V(x)|
\le
\frac{\varepsilon}{1-\varepsilon}\,\hat V(x),
\end{equation}
where $V$ is the unique solution to \eqref{eq:lyap_pde} defined by \eqref{eq:lyap_integral}. 
\end{thm}

\begin{proof}
Let $\phi(t,x)$ denote the trajectory of \eqref{eq:sys} starting from $x$. 
Along trajectories, by the definition of \eqref{eq:lyap_residual}, 
\begin{align*}
\frac{d}{dt}\hat V(\phi(t,x))
& =
-\omega(\phi(t,x))+r(\phi(t,x)).
\end{align*}
Integrating over $[0,T]$ gives
\[
\hat V(\phi(T,x))-\hat V(x)
=
-\int_0^T \omega(\phi(t,x))\,dt
+
\int_0^T r(\phi(t,x))\,dt .
\]
By Lemma \ref{lem:lyap}, the true solution satisfies (\ref{eq:lyap_integral}). 
Since $\Omega\subset\mathcal{D}$, we have $\phi(t,x)\to0$ as $t\to\infty$. 
By continuity and $\hat V(0)=0$, 
$ 
\lim_{T\to\infty}\hat V(\phi(T,x))=0.
$
Letting $T\to\infty$ therefore yields
\begin{equation}\label{eq:improper_integral}
V(x) - \hat V(x)=\int_0^\infty r(\phi(t,x))\,dt .    
\end{equation}
Using $|r|\le\varepsilon\omega$ and $\omega\ge0$,
\begin{align*}
|\hat V(x)-V(x)|
& \le
\int_0^\infty |r(\phi(t,x))|\,dt \\ 
& \le
\varepsilon\int_0^\infty \omega(\phi(t,x))\,dt
= 
\varepsilon V(x),
\end{align*}
which proves \eqref{eq:relative_error}. The bound  
\eqref{eq:posterior_error} follows immediately. 
\end{proof}

\begin{rem}
The residual condition \eqref{eq:relative_residual_lyap} is expressed relative to $\omega(x)$ rather than as a uniform bound on $|r(x)|$. This is essential for obtaining the state-dependent bounds \eqref{eq:relative_error} and \eqref{eq:posterior_error}. The following example shows that an absolute bound $|r(x)|\le \delta$ does not necessarily yield a finite trajectory-integral needed for bounding $|\hat V(x)-V(x)|$.

Consider $\dot x=-x$ on $\Omega=(-\tfrac12,\tfrac12)$ with $\omega(x)=x^2$. Then $\phi(t,x)=xe^{-t}$ and $V(x)=\int_0^\infty x^2e^{-2t}dt=\tfrac12x^2$. Let
\[
r(x)=
\begin{cases}
\displaystyle \frac{\delta\log 2}{|\log|x||}, & x\neq 0,\\[0.6ex]
0, & x=0.
\end{cases}
\]
Then $r$ is continuous, $r(0)=0$, and $|r(x)|\le\delta$ on $\Omega$. However, for $x\neq0$,
\[
\int_0^\infty r(\phi(t,x))\,dt
=
\delta\log 2\int_0^\infty \frac{dt}{t+|\log|x||}
=
\infty.
\]
Thus a continuous residual with arbitrarily small uniform bound does not necessarily provide a finite bound on $|\hat V(x)-V(x)|$. This motivates the relative residual condition \eqref{eq:relative_residual_lyap}. \hfill $\Box$
\end{rem}

We also state the following one-sided consequence of Theorem~\ref{thm:lyap_error}, which is of particular interest for the construction and verification of Lyapunov functions.

\begin{cor}
\label{cor:onesided}
Let Assumption~\ref{ass:lyap} hold and let $\hat V\in C^1(\Omega)$ satisfy
$\hat V(0)=0$. Assume that the residual defined in \eqref{eq:lyap_residual} satisfies
\begin{equation}\label{eq:relative_residual_cor}
r(x)\le \varepsilon\,\omega(x)\qquad \forall x\in\Omega,
\end{equation}
for some $\eps\in [0,1)$, and $\omega(x)$ is positive definite. Then for all $x\in\Omega$,
\begin{equation}\label{eq:lyap_one_side_bound}
\hat V(x)\ge (1-\varepsilon)V(x),
\end{equation}
\begin{equation}\label{eq:lyap_inequality}
D\hat V(x)\cdot f(x)\le -(1-\varepsilon)\,\omega(x),    
\end{equation}
and, by construction, $\hat V$ is positive definite on $\Omega$. Hence $\hat V$ is a Lyapunov function for \eqref{eq:sys} on $\Omega$.
\end{cor}

\begin{proof}
The inequality $\hat V(x)\ge (1-\varepsilon)V(x)$ follows directly from the proof of 
Theorem~\ref{thm:lyap_error}, in particular from \eqref{eq:improper_integral}. From the definition of the residual \eqref{eq:lyap_residual} and the bound \eqref{eq:relative_residual_cor}, we obtain \eqref{eq:lyap_inequality}. Since $\omega$ is positive definite, we can show that
\[
V(x)=\int_0^\infty \omega(\phi(t,x))\,dt
\]
is positive definite. By \eqref{eq:lyap_one_side_bound}, $\hat V$ is also positive definite. Furthermore,
$D\hat V(x)\cdot f(x)$ is negative definite by \eqref{eq:lyap_inequality}. We conclude that $\hat V$ is a Lyapunov function for \eqref{eq:sys} on $\Omega$.
\end{proof}

\begin{rem}
\label{rem:sublevel}
If $\Omega$ is not known to be forward invariant, one may instead work on a
sublevel set of $\hat V$. Let
$$
\Omega_c:=\{x\in\Omega:\hat V(x)\le c\},
$$
and assume that the residual bound \eqref{eq:relative_residual_cor} holds on
$\Omega_c$ and that
$
\overline{\Omega_c}\cap \partial\Omega=\varnothing.
$
Then $\Omega_c$ is forward invariant. Indeed, by \eqref{eq:lyap_inequality},
$D\hat V(x)\cdot f(x)<0$ whenever $\hat V(x)=c$. Hence a 
trajectory starting from $\Omega_c$ cannot leave $\Omega_c$ through the level
set $\{\hat V=c\}$. Since $\overline{\Omega_c}\cap\partial\Omega=\varnothing$,
it also cannot reach the boundary of $\Omega$ while remaining in $\Omega_c$.
Therefore $\Omega_c$ is forward invariant, and the conclusions of
Theorem~\ref{thm:lyap_error} and Corollary~\ref{cor:onesided} remain valid on
$\Omega_c$. \hfill $\Box$
\end{rem}

\section{Error bounds for the HJB equation}

We derive solution error bounds for the Hamilton--Jacobi--Bellman (HJB) equation from residual errors of an approximate value function. 
Such bounds are useful not only because they provide control Lyapunov certificates \cite{liu2025formally}, but also because they rigorously quantify optimality gaps.

Consider the control-affine system~\eqref{eq:control_sys} and the associated HJB equation~\eqref{eq:hjb_reduced}.  
Let $\hat V\in C^1(\Omega)$ be an approximate value function for solving (\ref{eq:hjb_reduced}). Define the HJB residual of $\hat V$ by 
\begin{equation}
\label{eq:hjb_residual}
\begin{aligned}
    r(x)
&:=
Q(x) + D\hat V(x)\cdot f(x) \\
&\quad -\tfrac14\, D\hat V(x)\,g(x)\,R(x)^{-1}\,g(x)^{T}\,D\hat V(x)^{T}. 
\end{aligned}
\end{equation}

\begin{ass}
\label{ass:hjb}
Assume that $\Omega\subset\mathbb R^n$ contains the origin in its interior. Furthermore, assume that the following hold:
\begin{enumerate}
    \item The functions $f$ and $g$ are locally Lipschitz. %
    \item The function $Q:\Omega\to\Real_{\ge 0}$ is positive definite with respect to the origin, and for every  $\rho>0$ there exists $\delta>0$ such that
    \(
    Q(x)\ge \delta
    \)
    for all $x\in \Omega\setminus B_\rho$, 
    where $B_\rho:=\{x\in\mathbb{R}^n:\,|x|<\rho\}$. 
    \item The function $R:\Omega\to\Real^{k\times k}$ is locally Lipschitz and $R(x)$ is positive definite for all $x\in\Omega$.
\end{enumerate}
\end{ass}

\begin{thm}\label{thm:hjb_error}
Let Assumption~\ref{ass:hjb} hold, and let $\eps\in[0,1)$.  
Let $\hat V\in C^1(\Omega)$ satisfy $\hat V(0)=0$. Assume that $D\hat V$ is locally Lipschitz on $\Omega$ and $\hat V$ is locally positive definite with respect to the origin and bounded from below on 
$$
\Omega_c:=\{x\in\Omega:\hat V(x)\le c\},
$$
where $\Omega_c$ also satisfies 
$
\overline{\Omega_c}\cap \partial\Omega=\varnothing.
$
Furthermore, assume that the residual of \(\hat V\) defined in \eqref{eq:hjb_residual} satisfies 
\begin{equation}
\label{eq:relative_residual_hjb}
\abs{r(x)}\le \eps Q(x)\qquad \forall x\in \Omega_c.   
\end{equation}
Then for all $x\in\Omega_c$,
\begin{equation}\label{eq:hjb_relative_error}
|\hat V(x)-V^*(x)|\le \varepsilon\,V^*(x),
\end{equation}
\begin{equation}\label{eq:hjb_posterior_error}
|\hat V(x)-V^*(x)|
\le
\frac{\varepsilon}{1-\varepsilon}\,\hat V(x).
\end{equation}
Furthermore, the feedback law
\begin{equation}
\label{eq:control_law}
\hat u(x)=-\frac12 R(x)^{-1}g(x)^\top D\hat V(x)^{T},
\end{equation}
satisfies the optimality gap bound
\begin{equation}
\label{eq:hjb_policy_gap}
0\le J(x,\hat u)-V^*(x)\le \frac{2\varepsilon}{1-\varepsilon}\,V^*(x) \le \frac{2\varepsilon}{(1-\varepsilon)^2}\,\hat V(x).
\end{equation}
\end{thm}

\begin{proof}
It can be verified by Assumption \ref{ass:hjb} that $\hat u$ is locally Lipschitz. Furthermore, since $\hat V$ is locally positive definite with respect to the origin and $\hat V\in C^1$, we must have $D\hat V(0)=0$. Hence $x=0$ is an equilibrium point for system (\ref{eq:control_sys}) with the controller $\hat u$. 

The HJB residual of \(\hat V\) defined in \eqref{eq:hjb_residual} can be written as 
\[
r(x)
=
Q(x)+\hat u(x)^\top R(x)\hat u(x)
+D\hat V(x)\cdot\bigl(f(x)+g(x)\hat u(x)\bigr).
\]

By the residual bound $r(x)\le \eps Q(x)$, the Lie derivative of $\hat V$ along the closed-loop vector field
\begin{equation}\label{eq:closed_loop_u_hat_pf}
f_{\mathrm{cl}}(x):=f(x)+g(x)\hat u(x)
\end{equation}
satisfies
\begin{align}
D\hat V(x)\cdot f_{\mathrm{cl}}(x)
&=
-\,Q(x)-\hat u(x)^\top R(x)\hat u(x)+r(x)\notag\\
&\le
-(1-\eps)Q(x)\label{eq:lyap_closed_loop_pf}
\end{align}
for all $x\in\Omega_c$. Since $\hat V$ is locally positive definite and $Q$ is positive definite, by a standard Lyapunov argument there exists $\rho>0$ such that every closed-loop trajectory generated by $\hat u$ starting in $B_\rho$ 
converges to the origin. 

We next show that every closed-loop trajectory generated by $\hat u$ starting in $\Omega_c$ converges to the origin. Let $\phi(t):=\phi(t,x,\hat u)$ with $x\in\Omega_c$. By the argument in Remark~\ref{rem:sublevel}, $\Omega_c$ is forward invariant for the closed-loop vector field \eqref{eq:closed_loop_u_hat_pf}. Suppose, to the contrary, that $\phi(t)$ does not converge to the origin. Then, by the local stability property above, $\phi(t)$ can never enter $B_\rho$, so
\(
|\phi(t)|\ge \rho 
\)
for all $t\ge 0$. By Assumption~\ref{ass:hjb}, there exists $\delta>0$ such that
\(
Q(\phi(t))\ge \delta 
\)
for all $t\ge 0$. 
On the other hand, by \eqref{eq:lyap_closed_loop_pf},
\begin{align*}
\frac{d}{dt}\hat V(\phi(t))
&\le
-(1-\eps)Q(\phi(t)).
\end{align*}
Integrating over $[0,T]$ yields
\begin{align*}
\hat V(x)-\hat V(\phi(T))
&\ge
(1-\eps)\int_0^T Q(\phi(t))\,dt\ge
(1-\eps)\delta T.
\end{align*}
Letting $T\to\infty$, we obtain $\hat V(\phi(T))\to -\infty$, which contradicts the assumption that $\hat V$ is bounded below on $\Omega_c$. Hence every closed-loop trajectory starting in $\Omega_c$ must enter $B_\rho$ in finite time and therefore converges to the origin.

We proceed to bound the solution errors using the residual error (\ref{eq:hjb_residual}). Note that both (\ref{eq:hjb_relative_error}) and (\ref{eq:hjb_posterior_error}) are equivalent to 
\begin{align}
\hat V(x) & \ge (1-\eps)V^*(x),\label{eq:lower_bound}\\
\hat V(x) & \le (1+\eps)V^*(x).\label{eq:upper_bound}
\end{align}

Using again the residual bound $r(x)\le \eps Q(x)$, along $\phi(t)=\phi(t,x,\hat u)$ we have
\begin{align*}
\frac{d}{dt}\hat V(\phi(t))
&=
D\hat V(\phi(t))\cdot
\bigl(f(\phi(t))+g(\phi(t))\hat u(\phi(t))\bigr) \\
&=r(\phi(t))-Q(\phi(t))-\hat u(\phi(t))^\top R(\phi(t))\hat u(\phi(t))\\
&\le
-(1-\eps)Q(\phi(t))
-\hat u(\phi(t))^TR(\phi(t))\hat u(\phi(t)) \\
&\le
-(1-\eps)
\Bigl(
Q(\phi(t))
+\hat u(\phi(t))^TR(\phi(t))\hat u(\phi(t))
\Bigr).
\end{align*}
Integrating over $[0,T]$ gives
\begin{align*}
&\hat V(x)-\hat V(\phi(T))\\
&\qquad\ge
(1-\eps)\int_0^T
\Bigl(
Q(\phi(t))
+\hat u(\phi(t))^TR(\phi(t))\hat u(\phi(t))
\Bigr)\,dt.
\end{align*}
Since $\phi(t)\to 0$ and $\hat V(0)=0$, letting $T\to\infty$ yields
\begin{align*}
\hat V(x)
&\ge
(1-\eps)\int_0^\infty
\Bigl(
Q(\phi(t))
+\hat u(\phi(t))^TR(\phi(t))\hat u(\phi(t))
\Bigr)\,dt \\
&=
(1-\eps)\,J(x,\hat u),
\end{align*}
where $J(x,\hat u)$ is the infinite-horizon cost associated with the feedback control $\hat u$. In particular, $\hat u$ is admissible, and therefore (\ref{eq:lower_bound}) follows from
\begin{align}
V^*(x)=\inf_{u\in\mathcal U}J(x,u)
\le
J(x,\hat u)
\le
\frac{1}{1-\eps}\hat V(x).
\label{eq:control_cost_estimate}
\end{align}

We now prove the bound (\ref{eq:upper_bound}). Fix $x\in\Omega_c$, let $u$ be any admissible control, and write
\(
\phi(t):=\phi(t,x,u).
\)
Define the last exit time
\[
\tau_x
:=
\sup\bigl\{\,t\ge 0:\hat V(\phi(t))>\hat V(x)\bigr\},
\]
with the convention $\tau_x=0$ if the trajectory never leaves the sublevel set
\[
\Omega_{\hat V(x)}:=\{y\in\Omega:\hat V(y)\le \hat V(x)\}.
\]
Then for all $t\ge \tau_x$,
\(
\hat V(\phi(t))\le \hat V(x),
\)
so
\(
\phi(t)\in\Omega_{\hat V(x)}\subseteq\Omega_c, 
\)
for all $t\ge \tau_x$. 
Moreover, by continuity of $t\mapsto \hat V(\phi(t))$, we have
\(
\hat V(\phi(\tau_x))=\hat V(x). 
\)

For $t\ge \tau_x$, by the definition of the residual \eqref{eq:hjb_residual},
\begin{align*}
&\frac{d}{dt}\hat V(\phi(t))\\
&\quad=
D\hat V(\phi(t))\cdot
\bigl(f(\phi(t))+g(\phi(t))u(t)\bigr)\\
&\quad=
-\,Q(\phi(t))\\
&\qquad
+\tfrac14 D\hat V(\phi(t))g(\phi(t))R(\phi(t))^{-1}
g(\phi(t))^TD\hat V(\phi(t))^T\\
&\qquad
+r(\phi(t))
+D\hat V(\phi(t))g(\phi(t))u(t).
\end{align*}
Completing the square gives
\begin{align*}
\frac{d}{dt}\hat V(\phi(t))
&=
-\,Q(\phi(t))
-u(t)^TR(\phi(t))u(t)\\
&\quad
+\eta(t)^TR(\phi(t))\eta(t)
+r(\phi(t)),
\end{align*}
\[
\eta(t)
:=
u(t)+\tfrac12 R(\phi(t))^{-1}g(\phi(t))^TD\hat V(\phi(t))^T.
\]
Since $R$ is positive definite and $r(x)\ge -\eps Q(x)$ on $\Omega_c$, we obtain, for all $t\ge \tau_x$,
\begin{align*}
\frac{d}{dt}\hat V(\phi(t))
&\ge
-(1+\eps)\Bigl(
Q(\phi(t))+u(t)^TR(\phi(t))u(t)
\Bigr).
\end{align*}
Integrating over $[\tau_x,T]$ yields
\begin{align*}
&\hat V(\phi(\tau_x))-\hat V(\phi(T))\\
&\quad\le
(1+\eps)\int_{\tau_x}^T
\Bigl(
Q(\phi(t))+u(t)^TR(\phi(t))u(t)
\Bigr)\,dt.
\end{align*}
Since $u$ is admissible, $\phi(T)\to 0$ as $T\to\infty$, and thus $\hat V(\phi(T))\to \hat V(0)=0$. Using $\hat V(\phi(\tau_x))=\hat V(x)$ and letting $T\to\infty$ gives
\begin{align*}
\hat V(x)
&\le
(1+\eps)\int_{\tau_x}^\infty
\Bigl(
Q(\phi(t))+u(t)^TR(\phi(t))u(t)
\Bigr)\,dt \\
&\le
(1+\eps)\,J(x,u).
\end{align*}
Since this holds for every $u\in\mathcal U$, it follows that
\begin{align*}
\hat V(x)
\le
(1+\eps)\inf_{u\in\mathcal U}J(x,u)=(1+\eps)V^*(x),
\end{align*}
which gives (\ref{eq:upper_bound}). The optimality gap (\ref{eq:hjb_policy_gap}) for $\hat u$ follows from (\ref{eq:lower_bound})--(\ref{eq:control_cost_estimate}). 
\end{proof}

We also state a one-sided version of Theorem \ref{thm:hjb_error} to emphasize the implication for constructing control Lyapunov functions and verification via solving the HJB equation. 

\begin{cor}
\label{cor:hjb_onesided}
Let the assumptions of Theorem~\ref{thm:hjb_error} hold, except that (\ref{eq:relative_residual_hjb}) is replaced by the one-sided condition
\begin{equation}\label{eq:hjb_residual_onesided}
r(x)\le \eps Q(x)\qquad \forall x\in \Omega_c,
\end{equation}
for some $\eps\in[0,1)$. Then for all $x\in\Omega_c$,
\begin{equation}\label{eq:hjb_onesided_value}
\hat V(x)\ge (1-\eps)V^*(x)
\end{equation}
and the feedback law (\ref{eq:control_law})
leads to
\begin{equation}\label{eq:hjb_onesided_clf}
D\hat V(x)\cdot\bigl(f(x)+g(x)\hat u(x)\bigr)
\le -(1-\eps)Q(x).
\end{equation}
Hence $\hat V$ is a control Lyapunov function on
$\Omega_c$.
\end{cor}

\begin{proof}
The estimate \eqref{eq:hjb_onesided_clf} follows directly from
\eqref{eq:hjb_residual} and the bound \eqref{eq:hjb_residual_onesided}, exactly as in the proof of Theorem~\ref{thm:hjb_error}. The same argument shows that every closed-loop trajectory generated by $\hat u$ starting in $\Omega_c$ converges to the origin, so $\hat u$ is admissible on $\Omega_c$. Integrating \eqref{eq:hjb_onesided_clf} along the closed-loop trajectory and letting $T\to\infty$ then gives
\(
\hat V(x)\ge (1-\varepsilon)J(x,\hat u). 
\)
Since $V^*(x)\le J(x,\hat u)$, we obtain \eqref{eq:hjb_onesided_value}. Finally, positive definiteness of $\hat V$ follows from \eqref{eq:hjb_onesided_value} and the positive definiteness of $V^*$ by definition.  
This, together with the inequality \eqref{eq:hjb_onesided_clf}, shows that $\hat V$ is a Lyapunov function for the closed-loop system under the feedback law $\hat u$ on $\Omega_c$. As a result, it is also a control Lyapunov function on $\Omega_c$.
\end{proof}

\section{Practically verifiable conditions}

The residual conditions appearing in the previous sections,
\[
|r(x)|\le \varepsilon \omega(x)
\quad\text{or}\quad
|r(x)|\le \varepsilon Q(x),
\]
must hold on a region containing the origin. Directly verifying such inequalities is often difficult because both sides vanish at the equilibrium point. As a result, SMT solvers \cite{gao2013dreal} and neural network verifiers \cite{xu2021fast,wang2021beta} may encounter degenerate situations near the origin, where points arbitrarily close to the equilibrium appear as potential counterexamples.

To address this issue, we assume that the stage cost functions admit a local quadratic lower bound near the origin. That is, there exist $\alpha>0$ and $\rho>0$ such that
\begin{equation}
\label{eq:alpha}
\omega(x)\ge \alpha\|x\|^2,
\qquad
Q(x)\ge \alpha\|x\|^2,
\qquad
\forall x\in B_\rho.
\end{equation}
Then it suffices to verify
\(
|r(x)|\le \varepsilon \alpha \|x\|^2. 
\)
We establish this inequality using a Hessian-type condition.

\begin{prop}
\label{prop:hessian_residual}
Let $\mathcal S\subset \mathbb{R}^n$ be a domain that is star-shaped with respect to the origin, and let $r\in C^2(\mathcal S)$ satisfy
\(
r(0)=0
\)
and 
\( 
Dr(0)=0. 
\)
Assume that
\begin{equation}
\|\nabla^2 r(x)\|_2\le 2\varepsilon\alpha
\qquad
\forall x\in \mathcal S. \label{eq:hessian_condition}   
\end{equation}
Then
\[
|r(x)|\le \varepsilon\alpha\|x\|^2
\qquad
\forall x\in \mathcal S .
\]
\end{prop}

\begin{proof}
Fix any $x\in \mathcal S$. Since $\mathcal S$ is star-shaped with respect to the origin, we have $tx\in \mathcal S$ for all $t\in[0,1]$. Taylor's theorem with integral remainder gives
\[
r(x)
=
r(0)+Dr(0)x+
\int_0^1 (1-t)\,x^\top \nabla^2 r(tx)\,x\,dt .
\]
Using $r(0)=0$ and $Dr(0)=0$, we obtain
\begin{align*}
r(x) &=
\int_0^1 (1-t)\,x^\top \nabla^2 r(tx)\,x\,dt,    \\
|r(x)|
&\le
\int_0^1 (1-t)\,\|\nabla^2 r(tx)\|_2\,\|x\|^2\,dt\\
&\le
\int_0^1 (1-t)\,2\varepsilon\alpha\,\|x\|^2\,dt
=
\varepsilon\alpha\|x\|^2 . 
\end{align*}
\end{proof}

\begin{rem}\label{rem:bias_correction}
To ensure that $r(0)=0$ and $Dr(0)=0$, we adopt a bias correction trick introduced in \cite{liu2025strict} by modifying $\hat V(x)$ to $\hat V(x) - D\hat V(0)\cdot x - \hat V(0)$. Then the modified $\hat V$ exactly satisfies $\hat V(0)=0$ and $D\hat V(0)=0$, which leads to  $r(0)=0$ and $Dr(0)=0$ for both residuals in (\ref{eq:lyap_residual}) and (\ref{eq:hjb_residual}). 
\end{rem}

\begin{rem}
\label{rem:hessian_condition_training}
We briefly comment on how the local Hessian condition in \eqref{eq:hessian_condition} may hold in practice. Standard neural-network training methods for Lyapunov and HJB PDEs, including \cite{zhou2024physics,meng2024physics}, often produce approximations with small relative residuals on the training and verification domain. For example, if the linearization at the origin is exponentially stable or stabilizable and the stage cost $\omega(x)$ or $Q(x)$ has a positive definite quadratic part, then the quadratic part of $\hat V$ can be trained to approximately satisfy the corresponding linear Lyapunov equation or algebraic Riccati equation. Hence $\nabla^2 r(0)$ can be made small. Since $r$ is $C^2$, continuity of $\nabla^2 r$ implies that \eqref{eq:hessian_condition} holds on a sufficiently small neighbourhood of the origin, such as $B_\rho$ with $\rho>0$ small enough.
\end{rem}

\begin{rem}
\label{rem:domain_split}
Based on Proposition \ref{prop:hessian_residual}, we can decompose $\Omega$ to verify  (\ref{eq:relative_residual_lyap}) or (\ref{eq:relative_residual_hjb}) as follows: Fix a radius $\rho>0$.
\begin{enumerate}
\item \textbf{Inner region.}
Verify that
\[
\|\nabla^2 r(x)\|_F \le 2\varepsilon\alpha,
\qquad
\forall x\in B_\rho .
\]
\item \textbf{Outer region.}
Verify the residual inequality directly: 
\[
|r(x)|\le \varepsilon\,\omega(x)
\quad\text{or}\quad
|r(x)|\le \varepsilon\,Q(x),
\quad
\forall x\in \Omega\setminus B_\rho .
\]
\end{enumerate}
Here, we use the Frobenius norm $|\nabla^2 r(x)|_F$ in (\ref{eq:hessian_condition}), since it is simpler to compute and upper bounds $|\nabla^2 r(x)|_2$.
\end{rem}

\section{Numerical examples}\label{sec:numerical_examples}

We consider an inverted pendulum on $\Omega=[-1,1]^2$.  
For the Lyapunov equation, we use the closed-loop dynamics
\[
\dot x_1=x_2,\;
\dot x_2=\sin(x_1)-x_2-(k_0x_1+k_1x_2),
\]
with $(k_0,k_1)=(4.4142,\,2.3163)$ and $\omega(x)=\|x\|^2$.  
For the HJB equation, we consider the control-affine system
\[
\dot x_1=x_2,\; 
\dot x_2=19.6\sin(x_1)-4x_2+40u,
\]
with $Q(x)=\|x\|^2$ and $R(x)=2$.  

In both cases, we train a one-hidden-layer extreme learning machine (ELM) \cite{huang2006extreme}  neural network $\hat V$ with 400 hidden units to approximately solve the corresponding Lyapunov or HJB PDE using the approach in \cite{zhou2024physics} and the tool LyZNet\footnote{\url{https://git.uwaterloo.ca/hybrid-systems-lab/lyznet/}, directory \texttt{examples/verify-pinn}} \cite{liu2024lyznet}. 
We chose this approach because it solves both Lyapunov and HJB PDEs accurately through linear least-squares optimization. 
Using a bias-corrected approximation (Remark~\ref{rem:bias_correction}) ensures $r(0)=0$ and $Dr(0)=0$.

We can verify the relative residual bounds \eqref{eq:relative_residual_lyap} and \eqref{eq:relative_residual_hjb} on $\Omega$ using $\alpha,\beta$-CROWN \cite{xu2021fast,wang2021beta}. More specifically, we split the domain as described in Remark~\ref{rem:domain_split} (with $\rho=0.1$), and verify the residual bounds for a candidate value of $\varepsilon$. We then perform a bisection search over $\varepsilon$ and report the smallest value for which verification succeeds. This gives $\varepsilon=3.3\times10^{-5}$ for the Lyapunov equation and $\varepsilon=2.2\times10^{-4}$ for the HJB example. In both cases, the verification takes around 120\,s on a single NVIDIA H100 NVL GPU (94\,GB HBM3 memory). Theorems~\ref{thm:lyap_error} and~\ref{thm:hjb_error} then yield certified \emph{a posteriori} error bounds for $\hat V$, shown in Figs.~\ref{fig:lyap_example} and~\ref{fig:hjb_example}. We further examine how the certified residual bound changes with network size; the results are reported in Table~\ref{tab:sensitivity}. As expected, larger networks generally solve the PDEs more accurately and lead to improved certified residual bounds, although verification times increase. We also found that $\rho=0.1$ provides a good compromise: smaller values make outer-region verification harder, while larger values make inner-region Hessian verification more challenging. An optimal choice of $\rho$ could be found by a simple line search, but we do not pursue this optimization here due to the page limit.

We conclude by emphasizing that the main contribution of this paper is to derive verifiable residual-to-solution error bounds for Lyapunov and HJB PDEs. Practical certification of the pointwise residual bounds relies on state-of-the-art verification tools such as $\alpha,\beta$-CROWN and remains challenging in higher dimensions. Extending the verification procedure to higher-dimensional systems through compositional verification, structure-exploiting bounds, or problem-specific decompositions is an interesting direction for future work.

\begin{table}[t]
\centering
\caption{Certified residual bounds and verification times for different network sizes.}
\label{tab:sensitivity}
\begin{tabular}{c c c c c}
\toprule
\makecell{Hidden units} & \makecell{Lyapunov ($\eps$)} & \makecell{time (s)} & \makecell{HJB ($\eps$)} & \makecell{time (s)} \\
\midrule
100 & $1.6\times10^{-3}$ & 23 & $1.3\times10^{-3}$ & 27 \\
200 & $9.5\times10^{-5}$ & 32 & $2.2\times10^{-3}$ & 28 \\
400 & $3.3\times10^{-5}$ & 120 & $2.2\times10^{-4}$ & 132 \\
800 & $1.8\times10^{-5}$ & 487 & $1.4\times10^{-4}$ & 654 \\
\bottomrule
\end{tabular}
\end{table}

\begin{figure}[t]
\centering
\includegraphics[width=\linewidth]{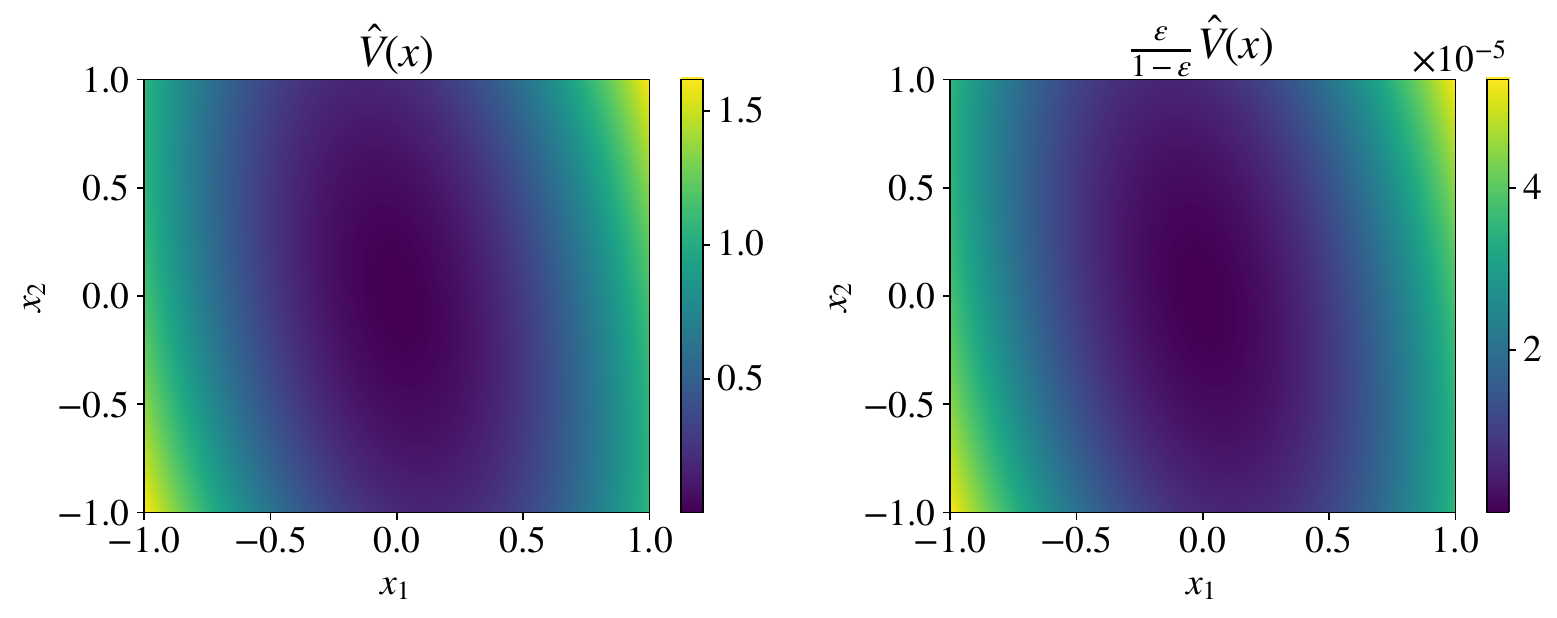}
\caption{
Left: neural approximation of the Lyapunov function.
Right: certified \emph{a posteriori} error bound.
}
\label{fig:lyap_example}
\end{figure}

\begin{figure}[t]
\centering
\includegraphics[width=\linewidth]{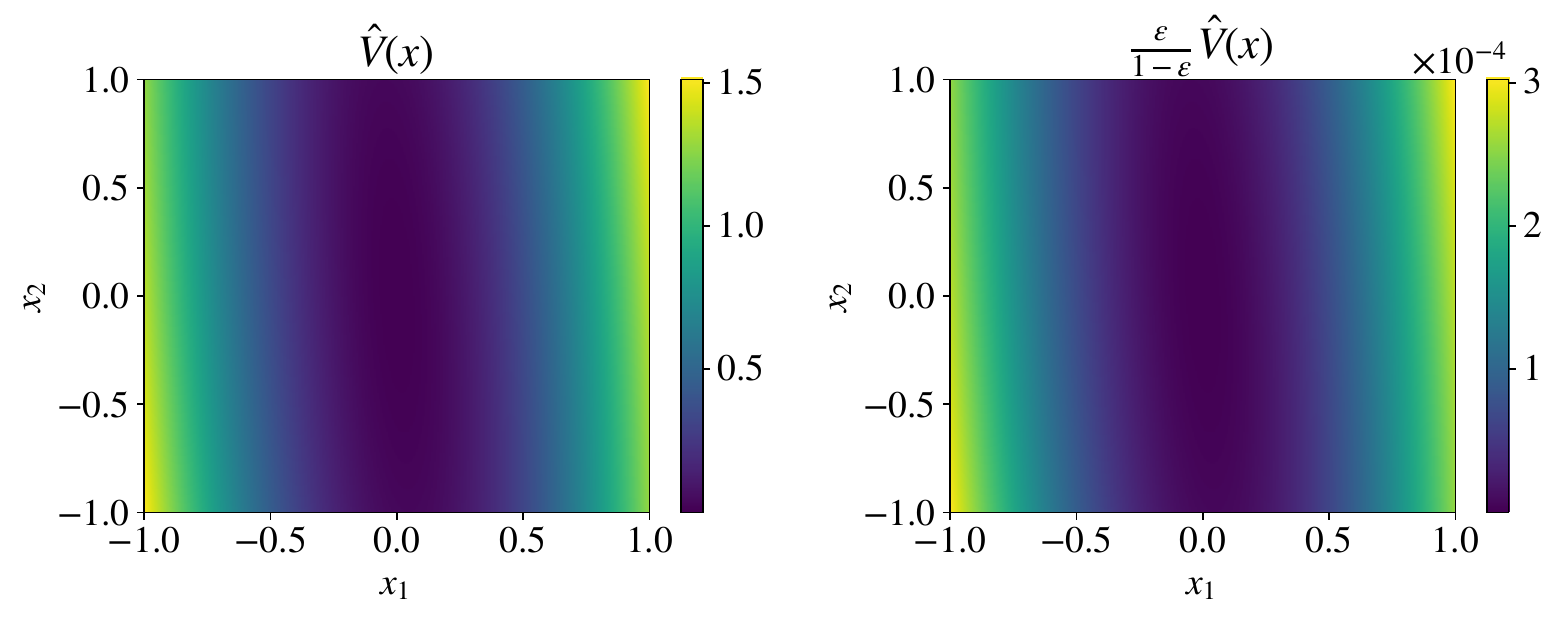}
\caption{
Left: neural approximation of the optimal value function.
Right: certified \emph{a posteriori} error bound.
}
\label{fig:hjb_example}
\end{figure}

\bibliographystyle{plain}
\bibliography{cdc26}

\begin{thebibliography}{10}

\bibitem{bardi1997optimal}
Martino Bardi and Italo Dolcetta-Capuzzo.
\newblock {\em Optimal Control and Viscosity Solutions of
  Hamilton-Jacobi-Bellman Equations}.
\newblock Springer, 1997.

\bibitem{beard1997galerkin}
Randal~W Beard, George~N Saridis, and John~T Wen.
\newblock Galerkin approximations of the generalized
  {H}amilton-{J}acobi-{B}ellman equation.
\newblock {\em Automatica}, 33(12):2159--2177, 1997.

\bibitem{camilli2001generalization}
Fabio Camilli, Lars Gr{\"u}ne, and Fabian Wirth.
\newblock A generalization of {Z}ubov's method to perturbed systems.
\newblock {\em SIAM Journal on Control and Optimization}, 40(2):496--515, 2001.

\bibitem{chang2019neural}
Ya-Chien Chang, Nima Roohi, and Sicun Gao.
\newblock Neural {L}yapunov control.
\newblock In {\em Proc. of NeurIPS}, volume~32, 2019.

\bibitem{fotiadis2025physics}
Filippos Fotiadis and Kyriakos~G. Vamvoudakis.
\newblock A physics-informed learning framework to solve the infinite-horizon
  optimal control problem.
\newblock {\em International Journal of Robust and Nonlinear Control}, 2025.

\bibitem{furfaro2022physics}
Roberto Furfaro, Andrea D'Ambrosio, Enrico Schiassi, and Andrea Scorsoglio.
\newblock Physics-informed neural networks for closed-loop guidance and control
  in aerospace systems.
\newblock In {\em AIAA SCITECH 2022 Forum}, page 0361, 2022.

\bibitem{gaby2022lyapunov}
Nathan Gaby, Fumin Zhang, and Xiaojing Ye.
\newblock Lyapunov-{N}et: A deep neural network architecture for {L}yapunov
  function approximation.
\newblock In {\em Proc. of CDC}, pages 2686--2691, 2022.

\bibitem{gao2013dreal}
Sicun Gao, Soonho Kong, and Edmund~M Clarke.
\newblock d{R}eal: An {SMT} solver for nonlinear theories over the reals.
\newblock In {\em Proc. of CADE}, 2013.

\bibitem{giesl2007construction}
Peter Giesl.
\newblock {\em Construction of Global {L}yapunov Functions using Radial Basis
  Functions}.
\newblock Springer, 2007.

\bibitem{grune2021computing}
Lars Gr{\"u}ne.
\newblock Computing {L}yapunov functions using deep neural networks.
\newblock {\em Journal of Computational Dynamics}, 8(2):131--152, 2021.

\bibitem{huang2006extreme}
Guang-Bin Huang, Qin-Yu Zhu, and Chee-Kheong Siew.
\newblock Extreme learning machine: theory and applications.
\newblock {\em Neurocomputing}, 70(1-3):489--501, 2006.

\bibitem{lagaris1998artificial}
Isaac~E Lagaris, Aristidis Likas, and Dimitrios~I Fotiadis.
\newblock Artificial neural networks for solving ordinary and partial
  differential equations.
\newblock {\em IEEE Transactions on Neural Networks}, 9(5):987--1000, 1998.

\bibitem{liu2025strict}
Jun Liu and Maxwell Fitzsimmons.
\newblock On strict verification of neural {L}yapunov functions.
\newblock {\em IFAC-PapersOnLine}, 59(19):304--309, 2025.

\bibitem{liu2025formally}
Jun Liu, Maxwell Fitzsimmons, Ruikun Zhou, and Yiming Meng.
\newblock Formally verified physics-informed neural control {L}yapunov
  functions.
\newblock In {\em Proc. of ACC}, pages 1347--1354. IEEE, 2025.

\bibitem{liu2023towards}
Jun Liu, Yiming Meng, Maxwell Fitzsimmons, and Ruikun Zhou.
\newblock Towards learning and verifying maximal neural {L}yapunov functions.
\newblock In {\em Proc. of CDC}, 2023.

\bibitem{liu2024lyznet}
Jun Liu, Yiming Meng, Maxwell Fitzsimmons, and Ruikun Zhou.
\newblock {LyZNet}: A lightweight python tool for learning and verifying neural
  {L}yapunov functions and regions of attraction.
\newblock In {\em Proc. of HSCC}, 2024.

\bibitem{liu2025physics}
Jun Liu, Yiming Meng, Maxwell Fitzsimmons, and Ruikun Zhou.
\newblock Physics-informed neural network {L}yapunov functions: {PDE}
  characterization, learning, and verification.
\newblock {\em Automatica}, 175:112193, 2025.

\bibitem{meng2024physics}
Yiming Meng, Ruikun Zhou, Amartya Mukherjee, Maxwell Fitzsimmons, Christopher
  Song, and Jun Liu.
\newblock Physics-informed neural network policy iteration: Algorithms,
  convergence, and verification.
\newblock In {\em Proc. of ICML}, 2024.

\bibitem{mitchell2005time}
Ian~M. Mitchell, Alexandre~M. Bayen, and Claire~J. Tomlin.
\newblock A time-dependent {H}amilton--{J}acobi formulation of reachable sets
  for continuous dynamic games.
\newblock {\em IEEE Transactions on Automatic Control}, 50(7):947--957, 2005.

\bibitem{raissi2019physics}
Maziar Raissi, Paris Perdikaris, and George~E Karniadakis.
\newblock Physics-informed neural networks: A deep learning framework for
  solving forward and inverse problems involving nonlinear partial differential
  equations.
\newblock {\em Journal of Computational Physics}, 378:686--707, 2019.

\bibitem{wang2021beta}
Shiqi Wang, Huan Zhang, Kaidi Xu, Xue Lin, Suman Jana, Cho-Jui Hsieh, and
  J~Zico Kolter.
\newblock {Beta-CROWN}: Efficient bound propagation with per-neuron split
  constraints for complete and incomplete neural network verification.
\newblock {\em Proc. of NeurIPS}, 34, 2021.

\bibitem{xu2021fast}
Kaidi Xu, Huan Zhang, Shiqi Wang, Yihan Wang, Suman Jana, Xue Lin, and Cho-Jui
  Hsieh.
\newblock Fast and complete: Enabling complete neural network verification with
  rapid and massively parallel incomplete verifiers.
\newblock In {\em Proc. of ICLR}, 2021.

\bibitem{zhou2024physics}
Ruikun Zhou, Maxwell Fitzsimmons, Yiming Meng, and Jun Liu.
\newblock Physics-informed extreme learning machine {L}yapunov functions.
\newblock {\em IEEE Control Systems Letters}, 8:1763--1768, 2024.

\bibitem{zhou2022neural}
Ruikun Zhou, Thanin Quartz, Hans De~Sterck, and Jun Liu.
\newblock Neural {L}yapunov control of unknown nonlinear systems with stability
  guarantees.
\newblock {\em Proc. of NeurIPS}, 35, 2022.

\end{thebibliography}

\end{document}